\documentclass[preprint,12pt]{elsarticle}

\usepackage{amssymb}
\usepackage{amsmath}
\usepackage{amsthm}
\usepackage{mathtools}

\usepackage[noend]{algpseudocode}
\usepackage{algorithm}

\usepackage{listings}
\usepackage{xcolor}
\usepackage[normalem]{ulem}
\usepackage{soul}
\usepackage{tikz}
\usetikzlibrary{trees, positioning, arrows.meta, calc, chains}

\definecolor{codegreen}{rgb}{0,0.6,0}
\definecolor{codegray}{rgb}{0.5,0.5,0.5}
\definecolor{codepurple}{rgb}{0.58,0,0.82}
\definecolor{backcolour}{rgb}{0.95,0.95,0.92}
\lstdefinestyle{mystyle}{
  backgroundcolor=\color{backcolour},
  commentstyle=\color{codegreen},
  keywordstyle=\color{magenta},
  numberstyle=\tiny\color{codegray},
  stringstyle=\color{codepurple},
  basicstyle=\ttfamily\footnotesize,
  breaklines=true,
  captionpos=b,
  numbers=left,
  numbersep=5pt,
  tabsize=2
}

\newtheorem{theorem}{Theorem}
\newtheorem{lemma}[theorem]{Lemma}
\newtheorem{example}[theorem]{Example}

\journal{Discrete Applied Mathematics}

\begin{document}

\begin{frontmatter}

\title{Generating cyclic pivot Gray codes for well-ordered $k$-degenerate graphs in constant amortized time}

\author[mpu]{Lei Dong}
\author[mpu]{Dennis Wong\corref{cor1}}
\author[mpu]{Bowie Liu}
\author[mpu]{Rui Bao}
\author[mpu]{Lin Chen}
\author[mpu]{Chan-Tong Lam}
\author[mpu]{Sio-Kei Im}

\address[mpu]{Macao Polytechnic University, Macao, China\\
\textit{email:} lei.dong@mpu.edu.mo, cwong@mpu.edu.mo, bowen.liu@mpu.edu.mo,
P2321323@mpu.edu.mo, lchen@mpu.edu.mo, ctlam@mpu.edu.mo, marcusim@mpu.edu.mo}

\cortext[cor1]{Corresponding author. Email: cwong@mpu.edu.mo}

\begin{abstract}
A graph $G$ is $k$-degenerate if there exists an ordering $v_1, v_2, \dots, v_n$ of its vertices such that each vertex $v_i$ has at most $k$ neighbors $v_j$ in $G$ with $j < i$.
A well-ordered $k$-degenerate graph is a labeled graph on the vertex set $\{1, 2, \dots, n\}$ in which every vertex $i$ has at most $k$ neighbors among $1, 2, \dots, i-1$.
We present the first simple algorithms that generate, rank, and unrank cyclic pivot Gray codes for well-ordered $k$-degenerate graphs, where consecutive graphs differ by the addition, removal, or pivoting of a single edge. Our algorithm generates each well-ordered $k$-degenerate graph in constant amortized time per graph, using $O(n^2)$ space, while ranking and unranking take $O(n^2)$ time and $O(n^2)$ space.
\end{abstract}

\begin{highlights}
\item First constant-amortized-time pivot Gray codes for well-ordered $k$-degenerate graphs.
\item Generation in $O(n^2)$ space; consecutive graphs differ by one edge insertion, deletion, or pivot.
\item Exact closed-form enumeration of well-ordered $k$-degenerate graphs.
\item Ranking and unranking in $O(n^2)$ time and $O(n^2)$ space.
\end{highlights}


\end{frontmatter}

\section{Introduction}

A graph $G$ is \emph{$k$-degenerate} (also known as \emph{$k$-inductive}) if there exists an ordering $v_1, v_2, \dots, v_n$ of its vertices such that each vertex $v_i$ has at most $k$ neighbors $v_j$ in $G$ with $j < i$.
Such an ordering is called an \emph{elimination ordering} or a \emph{degeneracy ordering}.
Equivalently, a graph $G$ is $k$-degenerate if and only if every subgraph of $G$ contains at least one vertex of degree at most $k$.
For example, Figure~\ref{k,n=4,k=2} provides several $2$-degenerate graphs on $n=4$ vertices, with an elimination ordering presented below each graph in the form of a tuple.
As an illustration, the elimination ordering for the rightmost graph in the first row of Figure~\ref{k,n=4,k=2} is $(2,4,1,3)$.
This corresponds to removing the vertices in reverse order (first $3$, then $1$, then $4$, and finally $2$), ensuring that at each step the removed vertex has degree at most two.

\begin{figure}[htbp]
\centering
\begin{tikzpicture}[
    scale=1.0,
    v/.style={circle, fill=black, inner sep=2pt},
    lab/.style={font=\scriptsize}
]

\begin{scope}[xshift=0cm, yshift=2.2cm]
    \node[v] (v1) at (0,1) {};
    \node[v] (v2) at (1,1) {};
    \node[v] (v3) at (1,0) {};
    \node[v] (v4) at (0,0) {};
    \draw (v1)--(v3);
    \draw (v1)--(v4);
    \draw (v3)--(v4);
    \node[lab, above left]  at (v1) {1};
    \node[lab, above right] at (v2) {2};
    \node[lab, below right] at (v3) {3};
    \node[lab, below left]  at (v4) {4};
    \node[lab, below] at (0.5,-0.45) {$(1,2,3,4)$};
\end{scope}

\begin{scope}[xshift=3.0cm, yshift=2.2cm]
    \node[v] (v1) at (0,1) {};
    \node[v] (v2) at (1,1) {};
    \node[v] (v3) at (1,0) {};
    \node[v] (v4) at (0,0) {};
    \draw (v1)--(v3);
    \draw (v2)--(v4);
    \node[lab, above left]  at (v1) {1};
    \node[lab, above right] at (v2) {3};
    \node[lab, below right] at (v3) {2};
    \node[lab, below left]  at (v4) {4};
    \node[lab, below] at (0.5,-0.45) {$(1,3,2,4)$};
\end{scope}

\begin{scope}[xshift=6.0cm, yshift=2.2cm]
    \node[v] (v1) at (0,1) {};
    \node[v] (v2) at (1,1) {};
    \node[v] (v3) at (1,0) {};
    \node[v] (v4) at (0,0) {};
    \draw (v1)--(v3);
    \draw (v2)--(v4);
    \draw (v3)--(v4);
    \node[lab, above left]  at (v1) {2};
    \node[lab, above right] at (v2) {1};
    \node[lab, below right] at (v3) {4};
    \node[lab, below left]  at (v4) {3};
    \node[lab, below] at (0.5,-0.45) {$(2,1,4,3)$};
\end{scope}

\begin{scope}[xshift=9.0cm, yshift=2.2cm]
    \node[v] (v1) at (0,1) {};
    \node[v] (v2) at (1,1) {};
    \node[v] (v3) at (1,0) {};
    \node[v] (v4) at (0,0) {};
    \draw (v1)--(v3);
    \draw (v1)--(v4);
    \node[lab, above left]  at (v1) {2};
    \node[lab, above right] at (v2) {4};
    \node[lab, below right] at (v3) {1};
    \node[lab, below left]  at (v4) {3};
    \node[lab, below] at (0.5,-0.45) {$(2,4,1,3)$};
\end{scope}

\begin{scope}[xshift=0cm, yshift=-0.5cm]
    \node[v] (v1) at (0,1) {};
    \node[v] (v2) at (1,1) {};
    \node[v] (v3) at (1,0) {};
    \node[v] (v4) at (0,0) {};
    \draw (v1)--(v3);
    \draw (v1)--(v4);
    \draw (v2)--(v4);
    \node[lab, above left]  at (v1) {3};
    \node[lab, above right] at (v2) {1};
    \node[lab, below right] at (v3) {2};
    \node[lab, below left]  at (v4) {4};
    \node[lab, below] at (0.5,-0.45) {$(3,1,2,4)$};
\end{scope}

\begin{scope}[xshift=3.0cm, yshift=-0.5cm]
    \node[v] (v1) at (0,1) {};
    \node[v] (v2) at (1,1) {};
    \node[v] (v3) at (1,0) {};
    \node[v] (v4) at (0,0) {};
    \draw (v1)--(v3);
    \draw (v1)--(v4);
    \draw (v2)--(v3);
    \node[lab, above left]  at (v1) {3};
    \node[lab, above right] at (v2) {4};
    \node[lab, below right] at (v3) {1};
    \node[lab, below left]  at (v4) {2};
    \node[lab, below] at (0.5,-0.45) {$(3,4,1,2)$};
\end{scope}

\begin{scope}[xshift=6.0cm, yshift=-0.5cm]
    \node[v] (v1) at (0,1) {};
    \node[v] (v2) at (1,1) {};
    \node[v] (v3) at (1,0) {};
    \node[v] (v4) at (0,0) {};
    \draw (v1)--(v3);
    \node[lab, above left]  at (v1) {3};
    \node[lab, above right] at (v2) {4};
    \node[lab, below right] at (v3) {2};
    \node[lab, below left]  at (v4) {1};
    \node[lab, below] at (0.5,-0.45) {$(3,4,2,1)$};
\end{scope}

\begin{scope}[xshift=9.0cm, yshift=-0.5cm]
    \node[v] (v1) at (0,1) {};
    \node[v] (v2) at (1,1) {};
    \node[v] (v3) at (1,0) {};
    \node[v] (v4) at (0,0) {};
    \draw (v1)--(v3);
    \draw (v2)--(v3);
    \node[lab, above left]  at (v1) {4};
    \node[lab, above right] at (v2) {1};
    \node[lab, below right] at (v3) {3};
    \node[lab, below left]  at (v4) {2};
    \node[lab, below] at (0.5,-0.45) {$(4,1,3,2)$};
\end{scope}

\end{tikzpicture}
\caption{Eight examples of $2$-degenerate graphs on $4$ vertices.}
\label{k,n=4,k=2}
\end{figure}
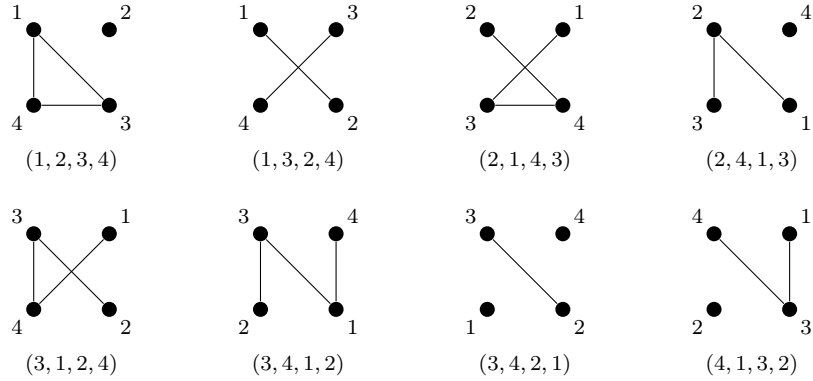

The study of $k$-degenerate graphs has a long history and has traditionally attracted considerable interest from mathematicians and graph theorists~\cite{bickle2012structural,chlebikova2007d,filakova1997note,lick1970k,mihok2001minimal}.
In particular, $1$-degenerate graphs are precisely forests: every non-empty subgraph has a leaf or an isolated vertex.
$2$-degenerate graphs include outerplanar graphs and series-parallel graphs; every outerplanar graph admits an embedding in which some vertex has degree at most $2$.
$k$-trees are maximal $k$-degenerate graphs, constructed by starting with a $(k+1)$-clique and iteratively adding a new vertex adjacent to all vertices of an existing $k$-clique.
$k$-degenerate graphs exhibit several useful structural and extremal properties. They are $(k+1)$-colorable, which follows from a greedy coloring along a degeneracy ordering; in particular, the chromatic number of such a graph is at most its degeneracy plus one. The size of the largest clique in a $k$-degenerate graph is at most $k+1$.
Degenerate graphs are closely related to $k$-core decomposition and its associated notions of coreness and degeneracy. These concepts have found numerous applications in social network analysis, bioinformatics, network visualization and resilience analysis, ecology, epidemic modeling (e.g., bootstrap percolation), and distributed computing. For a more comprehensive survey of applications and recent developments in this area, see~\cite{malliaros2019core}.

Enumerating $k$-degenerate graphs, however, is challenging in both the unlabeled and labeled settings.
While the unlabeled case is already difficult due to the need to identify isomorphic copies,
the labeled case presents a further subtlety: a single graph may admit many distinct elimination orderings.
If each ordering is treated as a separate construction, the same graph may be generated multiple times, leading to overcounting.
For example, each graph in Figure~\ref{k,n=4,k=2} is paired with a valid elimination ordering shown below it, yet every one of these graphs also admits $(1,2,3,4)$ as a valid elimination ordering.

Bauer, Krug, and Wagner~\cite{bauer2010enumerating}
(elaborated further in Krug's dissertation~\cite{krugcombinatorial})
therefore considered a labeled variant of $k$-degenerate graphs, known as \emph{well-ordered $k$-degenerate graphs},
which bypasses the ordering ambiguity by fixing the vertex labels as the elimination ordering.
In these graphs, vertices are labeled $1, 2, \dots,n$, and each vertex $i$ is adjacent to at most $k$ vertices with larger labels; reversing the vertex order shows that this is equivalent to our definition, namely, a labeled graph on the vertex set $\{1, 2, \dots, n\}$ in which every vertex $i$ has at most $k$ neighbors among $1, 2, \dots, i-1$.
Bauer et al.~\cite{bauer2010enumerating} studied the enumeration and generation of well-ordered $k$-degenerate graphs. They showed that the number $D_n^{(k)}$ of well-ordered $k$-degenerate graphs on $n$ vertices satisfies
$D_n^{(k)}
=
\sum_{i=0}^{\min\{n-1,k\}}
\binom{n-1}{i}
D_{n-1}^{(k)}$.
For example, when $k=1$, the numbers of well-ordered $k$-degenerate graphs for $n=1, 2 ,\ldots$ are
$1, 2, 6, 24, 120, \ldots$, which form OEIS sequence A000142~\cite{oeis:A000142}. When $k=2$, the corresponding sequence begins
$1,2,8,56,616,9856,\ldots$, which is OEIS sequence A128814~\cite{oeis:A128814}.

One of the most important aspects of combinatorial generation is to list the instances of a combinatorial object so that consecutive instances differ by a specified \emph{closeness condition} involving a constant amount of change.
Lists of this type are called \emph{Gray codes}.
This terminology is due to the eponymous \emph{binary reflected Gray code} (BRGC) by Frank Gray, which orders the $2^n$ binary strings of length $n$ so that consecutive strings differ by one bit.
For example, when $n = 4$ the order is
\begin{center}
\small
0000, 1000, 1100, 0100, 0110, 1110, 1010, 0010,\\ 0011, 1011, 1111, 0111, 0101, 1101, 1001, 0001.
\end{center}
Variations that reverse the entire order or the individual strings are also commonly used
in practice and in the literature.
For more information about binary reflected Gray code and its variations, see~\cite{10.1007/978-3-030-85088-3_15,flipswap,Vajnovszki2007GrayCO}.
We note that the order above is \emph{cyclic} because the last and first strings also differ by the closeness condition, and this property holds for all $n$.
The BRGC listing is a \emph{1-Gray code} in which consecutive strings differ by one symbol change.

When the combinatorial instances are graphs, consecutive graphs differ by a specified closeness condition involving a constant amount of change---a local structural operation on the edge set or vertex set. A long-standing problem in combinatorial generation is to determine whether fundamental graph classes admit Gray codes and, if so, how to generate them efficiently.
A well-known special case of this problem was posed as an open problem by Knuth~\cite{knuth2011art} (Exercise~101 in Section~7.2.1.6 of \emph{The Art of Computer Programming}, Volume~4), with a difficulty rating of 46 out of 50: Is there a simple revolving-door way to list all $n^{n-2}$ spanning trees of the complete graph $K_n$? (The order produced by Algorithm~S is quite complicated.)
In this direction, pivot Gray codes have attracted considerable attention in recent years. Cameron, Grubb, and Sawada~\cite{cameron2024pivot} first introduced pivot Gray codes (also known as strong revolving-door Gray codes~\cite{knuth2011art}) for spanning trees of graphs, where consecutive spanning trees differ by an edge exchange whose removed and added edges share a common vertex. They constructed such a code for the spanning trees of fan graphs and further posed the open problem of whether similar pivot Gray code constructions exist for other graph classes, including the $n$-cube and the complete graph.
Subsequently, Behrooznia and M\"utze~\cite{behrooznia2024listing} established pivot-exchange Gray codes for spanning trees of outerplanar graphs. They observed that Knuth's open problem of whether the complete graph $K_n$ admits a ``nice'' Gray code listing of its spanning trees can be interpreted, in the spirit of the problem of Cameron et al., as asking whether there exists a pivot-exchange Gray code for the spanning trees of $K_n$. Alternatively, such a Gray code could provide a more fine-grained explanation of Cayley's formula $n^{n-2}$ for the number of spanning trees.
Liu et al.~\cite{liu2026generating} further extended pivot Gray codes to spanning trees of complete graphs and achieved constant amortized time generation, providing the simple revolving-door listing sought in Knuth's Exercise~101; their result also answers the interpretation proposed by Behrooznia and M\"utze and settles the complete-graph case of the open problem posed by Cameron et al.
More recently, Bonamy et al.~\cite{bonamy2026gray} established the existence of a pivot Gray code for arborescences of tournaments, showing that Gray codes also exist in directed settings.
Despite these advances, the existence of pivot Gray code constructions for broader graph families remains open in general. Bauer et al.~\cite{bauer2010enumerating} presented efficient algorithms for generating and enumerating well-ordered $k$-degenerate graphs. However, the resulting enumeration is given recursively rather than in closed form and does not exhibit the regularity characteristic of Gray code orderings.
Motivated by this, we introduce a generalized framework of pivot Gray codes for well-ordered $k$-degenerate graphs by extending the traditional pivot operation to allow a single edge insertion or deletion. We still refer to this extended operation as a pivot operation. A Gray code is \emph{cyclic} if the first and last graphs in the sequence differ by exactly one pivot operation. Within this framework, we construct cyclic Gray codes for the entire graph class. Our result extends pivot Gray code constructions from spanning trees to a broader class of graphs.

The rest of the paper is organized as follows. In Section~2, we present a recursive algorithm for generating well-ordered $k$-degenerate graphs in constant amortized time. Section~3 derives an exact enumeration formula for this class of graphs. In Section~4, we develop efficient ranking and unranking algorithms to determine the position of a given graph in the generation order and, conversely, to recover the graph corresponding to a given position.

\section{Pivot Gray Codes for Well-Ordered $k$-Degenerate Graphs}
\label{sec:gen}

In this section, we present the first cyclic pivot Gray code for well-ordered $k$-degenerate graphs, along with a simple recursive algorithm that generates it.
The recursion relies on the generation of Gray codes for binary strings whose weight (i.e., the number of 1s) is at most $k$.

This is achieved by generating the binary reflected Gray code (BRGC) for length-$n$ strings via the construction described in \cite{flipswap}; let $\text{BRGC}(n)$ denote the binary reflected Gray code listing of length $n$, defined recursively by $\text{BRGC}(1) = 0, 1$ and, for $n > 1$, $\text{BRGC}(n) = (\text{BRGC}(n-1) \cdot 0,\; \overline{\text{BRGC}}(n-1) \cdot 1)$, where $L \cdot x$ appends $x$ to each string of a listing $L$, and $\overline{L}$ reverses it. This is a \emph{global recursion}, in that it refers to the entire list $\text{BRGC}(n-1)$ as one unit. It is equivalent to the following \emph{local recursion}, which refers to the individual strings: if $\text{BRGC}(n-1) = b_1, b_2, \ldots, b_{2^{n-1}}$, then $\text{BRGC}(n) = 0 \cdot b_1,\; 1 \cdot b_1,\; 1 \cdot b_2,\; 0 \cdot b_2,\; \ldots,\; 0 \cdot b_{2^{n-1}-1},\; 1 \cdot b_{2^{n-1}-1},\; 1 \cdot b_{2^{n-1}},\; 0 \cdot b_{2^{n-1}}$, with $\text{BRGC}(1) = 0, 1$; that is, fix the bits one by one from the last to the first, and at each step alternate between appending $0$ then $1$ and $1$ then $0$ to successive strings. This equivalence between the global recursion and the local recursion is guaranteed by Williams~\cite{williams2013greedy}, whose \emph{global recursion} is $\mathrm{Gray}(1) = 0, 1$ and, for $n > 1$, $\mathrm{Gray}(n) = (0 \cdot \mathrm{Gray}(n-1),\; 1 \cdot \overline{\mathrm{Gray}}(n-1))$, and whose \emph{local recursion} is $\mathrm{Gray}(n) = b_1 \cdot 0,\; b_1 \cdot 1,\; b_2 \cdot 1,\; b_2 \cdot 0,\; \ldots,\; b_{2^{n-1}-1} \cdot 0,\; b_{2^{n-1}-1} \cdot 1,\; b_{2^{n-1}} \cdot 1,\; b_{2^{n-1}} \cdot 0$, with $\mathrm{Gray}(1) = 0, 1$; in his \emph{global recursion} the new bit is the first bit and, correspondingly, in his \emph{local recursion} the new bit is appended to the end and the bits are fixed from the first onwards, that is, in the opposite direction to ours, but the logic is identical. During the generation, any string whose weight exceeds $k$ is pruned. We call the resulting listing the \emph{pruned BRGC}. Because the set of strings with weight at most $k$ forms a flip-swap language \cite{flipswap}, by Theorem~2 therein the pruned BRGC is a 2-Gray code in which consecutive strings differ in at most two bits.

In fact, Algorithm~\textsc{PruneBRGC} can generate this Gray code directly. The recursion traverses the bits one by one, from the $n$-th bit to the first, so the maximum depth is $n$. While the current weight is smaller than $k$, the recursion keeps generating downward: it fixes the current bit to one of the two values and continues to the next bit, leaving the other value to be continued later; once the weight reaches $k$, the remaining bits are set to $0$ and the generation along this path terminates. Each branching point triggers a new recursive call, and every generated branch eventually produces a valid string, so the number of recursive calls is proportional to the number of generated strings. 
The execution involves no backtracking: the already determined suffix is never modified, and the algorithm maintains only one buffer of length $n$; per branching point it stores a constant number of integer parameters, namely its position, weight, and direction. When a branch is processed, the fixed suffix in the buffer is reused to continue generating the remaining bits.

This recursion follows the BRGC order, stops a path as soon as the weight reaches $k$, and below weight $k$ explores both values of the current bit; it therefore outputs a string exactly when its weight is at most $k$. Since the pruned BRGC is precisely the list of the strings of weight at most $k$ in the BRGC order, the two listings consist of the same strings in the same order, and the recursion generates exactly the pruned BRGC.

\begin{algorithm}[h]
\caption{Pruned BRGC for binary strings of length $n$ with weight at most $k$}
\label{alg:pruned-brgc}
\begin{algorithmic}[1]
\Function{PruneBRGC}{$n,k$}
    \State $b \gets 0^n$ \Comment{shared buffer}
    \State \Call{Gen}{$n-1,0,1$}
\EndFunction
\Function{Gen}{$i,w,d$} \Comment{$d\in\{1,-1\}$: order of the two branches}
    \If{$i < 0$ \textbf{or} $w = k$}
        \State \textbf{output} $b$ \Comment{remaining bits are all $0$}
        \State \Return
    \EndIf
    \State $e \gets 1$
    \For{$x \in \{0,1\}$ in the order given by $d$} 
        \State $b_i \gets x$
        \State \Call{Gen}{$i-1,\,w+x,\,e$}
        \State $e \gets -e$
    \EndFor
    \State $b_i \gets 0$ \Comment{restore the current bit}
\EndFunction
\end{algorithmic}
\end{algorithm}

\begin{lemma}
\label{lem:brgc-rec}
The algorithm \textsc{PruneBRGC} generates all binary strings of length $n$ with weight at most $k$ in constant amortized time per string and uses $O(n)$ space.
\end{lemma}

\begin{proof}
As standard for generation algorithms, the time required to output a string is not included in the analysis.

The recursion traverses the bits one by one, from the $n$-th bit to the first, so the maximum recursion depth is $n$. During the traversal, when $\mathrm{weight}<k$, the current path splits into two branches, one explored immediately and the other later. When $\mathrm{weight}=k$, no new branch is created; instead, all remaining bits are set to $0$, and the generation of branches along this path terminates.

Whenever the algorithm encounters a branching point, it makes a new recursive call. Thus, each newly created branch corresponds to one new recursive call. Moreover, every generated branch eventually produces a valid string, and there are no invalid branches. Therefore, the number of recursive calls grows proportionally with the number of branches, and the total number of recursive calls is proportional to the total number $M$ of valid strings. Since each recursive call requires constant time, the total running time is $O(M)$. Hence, 
amortizing the total time $O(M)$ over the $M$ valid strings gives constant amortized time per string.

The recursion maintains only the current string, of length $n$, so the total space is $O(n)$.
\end{proof}

We illustrate Algorithm~\textsc{PruneBRGC} on $n=4$ and $k=2$. The algorithm fixes the bits from the last to the first: when the current weight is smaller than $2$, the current path splits into two branches, one explored immediately and the other later, and when the current weight reaches $2$, the remaining bits are forced to be $0$ and the completed string is output immediately. The resulting listing is

\begin{center}
0000, 1000, 1100, 0100, 0110, 1010, 0010,
0011, 0101, 1001, 0001,
\end{center}

which is exactly the BRGC of length $4$ with every string of weight larger than $2$ removed. Consecutive strings differ in at most two bits. For example, the string $0011$ is output as soon as its last two bits are fixed, since their weight already equals $2$; the two remaining bits are set to $0$ without further branching.

\begin{figure}[h]
\centering
\includegraphics[width=\columnwidth]{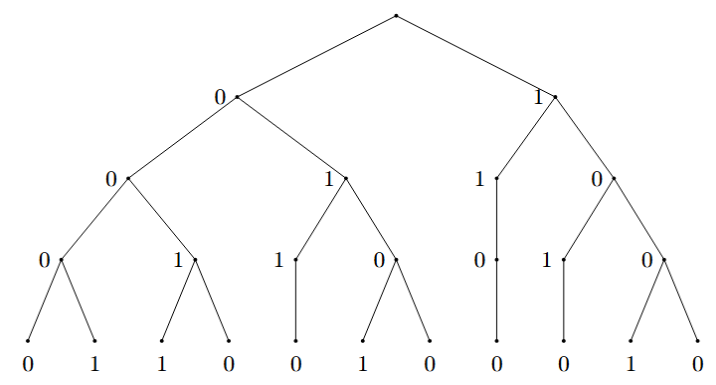} 
\caption{Recursive generation tree of the pruned BRGC for $n=4$ and $k=2$.}
\label{fig:brgc-tree}
\end{figure}

\medskip

For $0 \leq k \leq n-1$, let $\mathcal{S}_{n-1,k} = \{ S \subseteq \{1, 2, \dots, n-1\} \mid |S| \leq k \}$ and $\mathcal{B}_{n-1,k} = \{ b \in \{0,1\}^{n-1} \mid \mathrm{wt}(b) \leq k \}$, where $\mathrm{wt}(b)$ is the weight of $b$. It is straightforward to see that there exists a bijection between the two collections.

\smallskip
By Theorem~2 of \cite{flipswap}, the pruned BRGC of $\mathcal{B}_{n-1,k}$ is a $2$-Gray code. Under the natural bijection between $\mathcal{B}_{n-1,k}$ and $\mathcal{S}_{n-1,k}$, this induces the following Gray code on $\mathcal{S}_{n-1,k}$.

\begin{lemma}
\label{lem:gray}
There exists a Gray code for $\mathcal{S}_{n-1,k}$, in which any two consecutive subsets differ by adding, removing, or replacing a single element (i.e., simultaneously removing one element and adding a different one).
\end{lemma}

\begin{proof}
By Theorem~2 of \cite{flipswap}, the pruned BRGC of $\mathcal{B}_{n-1,k}$ is a 2-Gray code in which consecutive strings differ in at most two bits. Under the natural bijection between $\mathcal{B}_{n-1,k}$ and $\mathcal{S}_{n-1,k}$, consecutive strings differ by

1. adding a single element (a $0 \to 1$ bit flip), or

2. removing a single element (a $1 \to 0$ bit flip), or

3. replacing one element with another (simultaneously flipping a $1$ to $0$ and a $0$ to $1$).

Therefore, the pruned BRGC induces a Gray code on $\mathcal{S}_{n-1,k}$ in which consecutive subsets differ by adding, removing, or replacing a single element.
\end{proof}

Building upon Lemma~\ref{lem:gray}, we obtain the following result.

\begin{lemma}
\label{thm:graph-sequence}
Fix a well-ordered $k$-degenerate graph $H$ on $\{1,\dots,n-1\}$.
Adding vertex $n$ to $H$ with neighbor set $S \subseteq V(H)$, $|S|\le k$,
where $S$ is taken in the Gray code order of $\mathcal{S}_{n-1,k}$,
produces a pivot Gray code on $\{1,\dots,n\}$ in which consecutive graphs
differ by adding, removing, or pivoting a single edge incident to vertex $n$.
\end{lemma}

Based on the above observations, Algorithm~\textsc{GenWOD} recursively generates a cyclic pivot Gray code for well-ordered $k$-degenerate graphs, and it proceeds as follows. For the base case $n=1$, it outputs the trivial graph and returns. For $n \ge 2$, it first recursively generates a cyclic pivot Gray code for well-ordered $k$-degenerate graphs on $n-1$ vertices; each of these graphs is referred to as a \emph{base graph}. For each base graph $G_{n-1}$, the algorithm enumerates all subsets $S \subseteq \{1,\dots,n-1\}$ with $|S| \le k$ in Gray code order, listed in forward order for the first base graph, in reverse order for the second, and alternating thereafter. For each subset $S$, a graph on $n$ vertices is obtained by adding vertex $n$ to $G_{n-1}$ and connecting it to the vertices in $S$. Each graph is output as soon as it is produced, and no list of graphs is ever stored. By Theorem~\ref{th1}, this construction yields a cyclic pivot Gray code for graphs on $n$ vertices. A complete Python implementation is provided in Appendix~A.

\begin{algorithm}[h]
\caption{Cyclic pivot Gray code for well-ordered $k$-degenerate graphs on $n$ vertices}
\label{o1}
\begin{algorithmic}[1]
\Function{GenWOD}{$n,k$} 
    \If{$n=1$}
        \State \textbf{output} $(\{1\},\varnothing)$
        \State \Return
    \EndIf
    \State $d \gets 0$
    \For{$G \in \Call{GenWOD}{n-1,k}$} 
        \For{$S \in$ Gray code listing of $\mathcal{S}_{n-1,k}$ with direction $(-1)^d$}
            \State \textbf{output} $(\{1,\dots,n\},\, E(G) \cup \{(n,v): v \in S\})$
        \EndFor
        \State $d \gets 1-d$
    \EndFor
\EndFunction
\end{algorithmic}
\end{algorithm}

\begin{lemma}\label{od}
Let $n \geq 1$ and $0 \leq k \leq n-1$, and let $\lambda(n,k)$ denote the number of well-ordered $k$-degenerate graphs on $n$ labeled vertices. Then:

\begin{enumerate}
\item If $n = 1$, or $n \geq 2$ and $k = 0$, then $\lambda(n,k) = 1$.
\item If $n \geq 2$ and $k \geq 1$, then $\lambda(n,k)$ is even.
\end{enumerate}
\end{lemma}

\begin{proof}
If $n=1$, there is exactly one graph, namely the trivial graph consisting of a single isolated vertex, so $\lambda(1,k)=1$.

Now suppose $n\ge2$. If $k=0$, then by the definition of a well-ordered $0$-degenerate graph, every vertex $i$ has no neighbors among $1, 2, \dots, i-1$; hence no edge can occur, and the unique graph is the empty graph. Therefore $\lambda(n,0)=1$.

If $k\ge1$, consider the edge $(1,2)$. Since vertex $2$ may have at most $k$ neighbors among earlier vertices, there are exactly two possibilities:

\begin{enumerate}
\item the edge $(1,2)$ is present;
\item the edge $(1,2)$ is absent.
\end{enumerate}

For every graph in the first class, deleting the edge $(1,2)$ produces a unique graph in the second class. Conversely, for every graph in the second class, adding the edge $(1,2)$ produces a unique graph in the first class. Thus the two classes contain the same number of graphs, say $m$, so $\lambda(n,k)=2m$ is even.
\end{proof}

\begin{theorem}\label{th1}
The algorithm \textsc{GenWOD} generates a list of all well-ordered $k$-degenerate graphs on $n$ vertices in cyclic pivot Gray code order, where consecutive graphs differ by the addition, removal, or pivot of a single edge.
\end{theorem}

\begin{proof}
When $n=1$ or $n\ge 2$ with $k=0$, there is only one well-ordered $k$-degenerate graph, i.e., the $n$-vertex empty graph, which trivially forms a cyclic pivot Gray code.

Below we discuss the case $n\ge 2$ and $k\ge 1$, where, by Lemma~\ref{od}, the number of well-ordered $k$-degenerate graphs is even.
The proof is by induction on $n$.

For the induction base $n = 2$, recall that we are in the case $k \ge 1$.
Since $0 \le k \le n-1 = 1$, this forces $k = 1$.
A well-ordered $1$-degenerate graph on $\{1,2\}$ allows vertex $2$ to be adjacent
to a subset $S \subseteq \{1\}$ with $|S| \le 1$, yielding exactly two graphs:
$G_1 = (\{1,2\}, \varnothing)$  $(S = \varnothing)$ and
$G_2 = (\{1,2\}, \{(1,2)\})$  $(S = \{1\})$.
The transition from $G_1$ to $G_2$ adds the edge $(1,2)$, and from $G_2$ back to $G_1$
removes it, forming a cyclic pivot Gray code.
This establishes the base case.

For $n \geq 3$ and $k \ge 1$, assume that for $n-1$ there exists a cyclic Gray code
$L = [G_1, G_2, \ldots, G_m]$
for the set of well-ordered $k$-degenerate graphs on $n-1$ vertices, where consecutive graphs differ by the addition, removal, or pivot of exactly one edge. By Lemma~\ref{od}, $m$ is even. Each graph $G_i = (\{1,2,\ldots,n-1\}, E_i)$ satisfies the well-ordered $k$-degenerate condition: for each vertex $j=2,\ldots,n-1$, it has at most $k$ neighbors in $\{1,\ldots,j-1\}$.

We construct a cyclic pivot Gray code for well-ordered $k$-degenerate graphs on $n$ vertices as follows.

For each $G_i$ in the inductive cyclic Gray code $L$, define the family $\mathbf{F}_i = \{ G_i^{(S)} \mid S \in \mathcal{S}_{n-1,k} \}$, where $G_i^{(S)} = (\{1,2,\ldots,n\}, E_i \cup \{ (n, v) \mid v \in S \})$ and $\mathcal{S}_{n-1,k} = \{ S \subseteq \{1,2,\ldots,n-1\} \mid |S| \leq k \}$. We establish that $G_i^{(S)}$ is a well-ordered $k$-degenerate graph on $n$ vertices. Since $G_i$ is well-ordered $k$-degenerate on $n-1$ vertices, and vertex $n$ connects to at most $k$ preceding vertices (by $|S| \leq k$), all conditions are satisfied.

Within $\mathbf{F}_i$, by Lemma~\ref{thm:graph-sequence}, for consecutive subsets $S_j, S_{j+1} \in \mathcal{S}_{n-1,k}$, the graphs $G_i^{(S_j)}$ and $G_i^{(S_{j+1})}$ differ by adding, removing, or pivoting a single edge incident to $n$, thereby satisfying the Gray code condition.

To connect $\mathbf{F}_i$ and $\mathbf{F}_{i+1}$ for all $i=1,2,\ldots,m-1$, by Lemma~\ref{lem:gray}, $\mathcal{S}_{n-1,k}$ has a Gray code sequence $S_1, S_2, \ldots, S_p$ (where $p = \sum_{i=0}^{k} \binom{n-1}{i}$), and the reverse sequence $S_p, S_{p-1}, \ldots, S_1$ is also a Gray code. For $\mathbf{F}_i$ with odd $i$, use the forward sequence, so $\mathbf{F}_i = (G_i^{(S_1)}, G_i^{(S_2)}, \ldots, G_i^{(S_p)})$, ending with $S_p$. For $\mathbf{F}_{i+1}$ with even $i+1$, use the reverse sequence, so $\mathbf{F}_{i+1} = (G_{i+1}^{(S_p)},{}$
$G_{i+1}^{(S_{p-1})},{}$ $\ldots,{}$ $G_{i+1}^{(S_1)})$, starting with $S_p$. The transition from $G_i^{(S_p)}$ to $G_{i+1}^{(S_p)}$ differs by the change in edge between $G_i$ and $G_{i+1}$ (addition, removal, or pivot of one edge), while the edges incident to $n$ remain the same. The same reasoning applies when $i$ is even and $i+1$ is odd, with the roles of forward and reverse interchanged. Alternating the enumeration directions of the subsets $S$ across successive families $\mathbf{F}_i$ (e.g., forward for $\mathbf{F}_1$, reverse for $\mathbf{F}_2$, forward for $\mathbf{F}_3$, \ldots) ensures that each $\mathbf{F}_i$ is internally a Gray code and transitions to $\mathbf{F}_{i+1}$ maintain the Gray code property.

Finally, note that $m$ is even, so $\mathbf{F}_m = (G_{m}^{(S_p)}, G_{m}^{(S_{p-1})}, \ldots, G_{m}^{(S_1)})$, ending with $S_1$, while $\mathbf{F}_1 = (G_1^{(S_1)}, G_1^{(S_2)}, \ldots, G_1^{(S_p)})$, starting with $S_1$. The transition from $G_m^{(S_1)}$ to $G_1^{(S_1)}$ differs by the edge change between $G_m$ and $G_1$ in the underlying $(n-1)$-vertex graphs, which, by the inductive hypothesis, is the addition, removal, or pivot of a single edge, while the edges incident to $n$ remain the same.
Therefore, the entire sequence forms a cyclic pivot Gray code.

To ensure that all well-ordered $k$-degenerate graphs on $n$ vertices are covered, consider any such graph $G = (\{1,2,\ldots,n\}, E)$. The induced subgraph on $\{1,2,\ldots,n-1\}$ is a well-ordered $k$-degenerate graph $G_i$ on $n-1$ vertices, and the neighbors of vertex $n$ form $S = \{ v \in \{1,2,\ldots,n-1\} \mid (n,v) \in E \}$ with $|S| \leq k$, thus $G = G_i^{(S)}$. Since the construction of the Gray code includes all $G_i$ and all $S \subseteq \{1, 2, \ldots, n-1\}$, every $G$ appears in the sequence, ensuring that all well-ordered $k$-degenerate graphs on $n$ vertices are fully covered. Since each $G_i^{(S)}$ is uniquely determined by $G_i$ and $S$, there are no duplicates.
\end{proof}

As an example, Figure~\ref{fig2} illustrates the cyclic pivot Gray code for $n=3$, $k=2$ generated by Algorithm~\textsc{GenWOD}. For $n=1$, the algorithm outputs the trivial graph $(\{1\},\varnothing)$. For $n=2$, the pruned BRGC of length $1$, namely $0,1$, corresponds to the neighbor sets $\varnothing$ and $\{1\}$ for vertex $2$, which yields the two base graphs (empty and with edge $(1,2)$). For $n=3$, Algorithm~\textsc{PruneBRGC} generates the pruned BRGC of length $2$, namely $00,10,11,01$, corresponding to subsets $\varnothing,\{1\},\{1,2\},\{2\}$ of $\{1,2\}$. For the first base graph (no edge), these subsets are enumerated in forward order, producing graphs $1$ through $4$; for the second base graph (edge $(1,2)$), they are enumerated in reverse order, producing graphs $5$ through $8$. Each consecutive pair differs by adding, removing, or pivoting a single edge, as indicated by the Gray code property; moreover, graph $1$ is obtained from graph $8$ by removing the edge $(1,2)$, so the listing is cyclic.

\begin{figure}[t]
    \centering
    \begin{tikzpicture}[
        scale=0.95,
        v/.style={circle, fill=black, inner sep=2pt},
        edge/.style={},
        lab/.style={font=\scriptsize, inner sep=3pt},
        layerlabel/.style={font=\small, anchor=east, align=right}
    ]

    \node[layerlabel] at (-6.0, 5.0) {Layer $n=1$};
    \node[layerlabel] at (-6.0, 3.0) {Layer $n=2$};
    \node[layerlabel] at (-6.0, -0.5) {Layer $n=3$};

    \draw[dashed, black] (-8.0, 4.3) -- (6.0, 4.3);
    \draw[dashed, black] (-8.0, 2.0) -- (6.0, 2.0);

    \begin{scope}[yshift=5cm]
        \node[v] at (0,0) {};
        \node[lab, above=4pt] at (0,0) {1};
    \end{scope}

    \begin{scope}[yshift=3cm, xshift=-2cm]
        \node[v] at (0,0) {};
        \node[v] at (1.2,0) {};
        \node[lab, above=4pt] at (0,0) {1};
        \node[lab, above=4pt] at (1.2,0) {2};
    \end{scope}
    \begin{scope}[yshift=3cm, xshift=2cm]
        \node[v] (a) at (0,0) {};
        \node[v] (b) at (1.2,0) {};
        \node[lab, above=4pt] at (0,0) {1};
        \node[lab, above=4pt] at (1.2,0) {2};
        \draw[edge] (a) -- (b);
    \end{scope}

    \begin{scope}[yshift=0.5cm, xshift=-4.5cm]
        \node[v] at (0,0.8) {};
        \node[v] at (-0.8,0) {};
        \node[v] at (0.8,0) {};
        \node[lab, above=4pt] at (0,0.8) {1};
        \node[lab, left=4pt] at (-0.8,0) {2};
        \node[lab, right=4pt] at (0.8,0) {3};
    \end{scope}

    \begin{scope}[yshift=0.5cm, xshift=-1.5cm]
        \node[v] (a1) at (0,0.8) {};
        \node[v] at (-0.8,0) {};
        \node[v] (a3) at (0.8,0) {};
        \node[lab, above=4pt] at (0,0.8) {1};
        \node[lab, left=4pt] at (-0.8,0) {2};
        \node[lab, right=4pt] at (0.8,0) {3};
        \draw[edge] (a1) -- (a3);
    \end{scope}

    \begin{scope}[yshift=0.5cm, xshift=1.5cm]
        \node[v] (b1) at (0,0.8) {};
        \node[v] (b2) at (-0.8,0) {};
        \node[v] (b3) at (0.8,0) {};
        \node[lab, above=4pt] at (0,0.8) {1};
        \node[lab, left=4pt] at (-0.8,0) {2};
        \node[lab, right=4pt] at (0.8,0) {3};
        \draw[edge] (b1) -- (b3);
        \draw[edge] (b2) -- (b3);
    \end{scope}

    \begin{scope}[yshift=0.5cm, xshift=4.5cm]
        \node[v] at (0,0.8) {};
        \node[v] (c2) at (-0.8,0) {};
        \node[v] (c3) at (0.8,0) {};
        \node[lab, above=4pt] at (0,0.8) {1};
        \node[lab, left=4pt] at (-0.8,0) {2};
        \node[lab, right=4pt] at (0.8,0) {3};
        \draw[edge] (c2) -- (c3);
    \end{scope}

    \begin{scope}[yshift=-1.5cm, xshift=-4.5cm]
        \node[v] (d1) at (0,0.8) {};
        \node[v] (d2) at (-0.8,0) {};
        \node[v] (d3) at (0.8,0) {};
        \node[lab, above=4pt] at (0,0.8) {1};
        \node[lab, left=4pt] at (-0.8,0) {2};
        \node[lab, right=4pt] at (0.8,0) {3};
        \draw[edge] (d1) -- (d2);
        \draw[edge] (d2) -- (d3);
    \end{scope}

    \begin{scope}[yshift=-1.5cm, xshift=-1.5cm]
        \node[v] (e1) at (0,0.8) {};
        \node[v] (e2) at (-0.8,0) {};
        \node[v] (e3) at (0.8,0) {};
        \node[lab, above=4pt] at (0,0.8) {1};
        \node[lab, left=4pt] at (-0.8,0) {2};
        \node[lab, right=4pt] at (0.8,0) {3};
        \draw[edge] (e1) -- (e2);
        \draw[edge] (e2) -- (e3);
        \draw[edge] (e1) -- (e3);
    \end{scope}

    \begin{scope}[yshift=-1.5cm, xshift=1.5cm]
        \node[v] (f1) at (0,0.8) {};
        \node[v] (f2) at (-0.8,0) {};
        \node[v] (f3) at (0.8,0) {};
        \node[lab, above=4pt] at (0,0.8) {1};
        \node[lab, left=4pt] at (-0.8,0) {2};
        \node[lab, right=4pt] at (0.8,0) {3};
        \draw[edge] (f1) -- (f2);
        \draw[edge] (f1) -- (f3);
    \end{scope}

    \begin{scope}[yshift=-1.5cm, xshift=4.5cm]
        \node[v] (g1) at (0,0.8) {};
        \node[v] (g2) at (-0.8,0) {};
        \node[v] at (0.8,0) {};
        \node[lab, above=4pt] at (0,0.8) {1};
        \node[lab, left=4pt] at (-0.8,0) {2};
        \node[lab, right=4pt] at (0.8,0) {3};
        \draw[edge] (g1) -- (g2);
    \end{scope}

    \end{tikzpicture}
    \caption{Well-ordered $2$-degenerate graphs on 3 vertices ($n=3$, $k=2$) in cyclic pivot Gray code order. The first row contains graphs 1--4 (left to right), and the second row contains graphs 5--8 (left to right).}
    \label{fig2}
\end{figure}
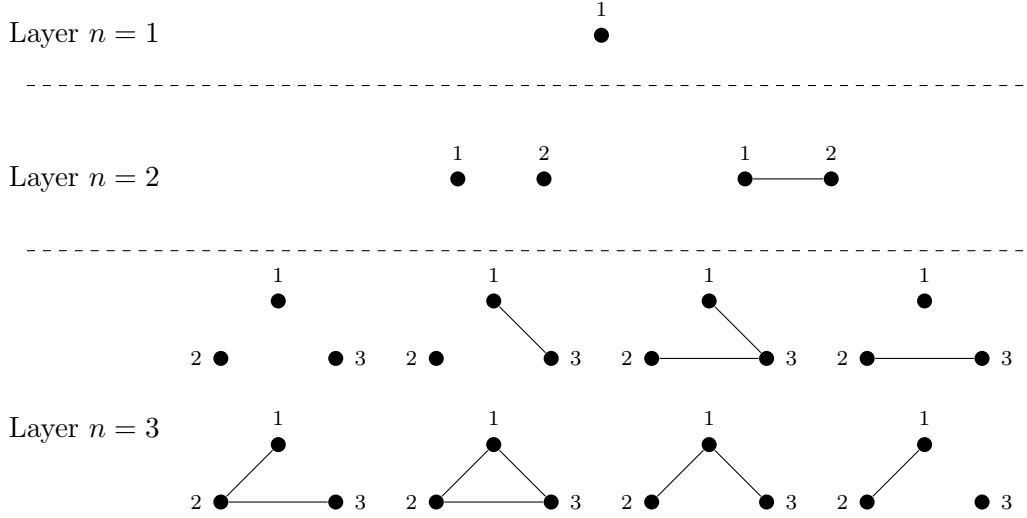

Applying Algorithm~\textsc{GenWOD} recursively to the 8 graphs on 3 vertices in Figure~\ref{fig2} yields the cyclic pivot Gray code for $n=4$ and $k=2$, shown in Figure~\ref{n=4,k=2}. For $n=4$, Algorithm~\textsc{PruneBRGC} generates the pruned BRGC of length $3$, namely $000, 100, 110, 010, 011, 101, 001$, which corresponds to the 7 subsets $\varnothing$, $\{1\}$, $\{1,2\}$, $\{2\}$, $\{2,3\}$, $\{1,3\}$, $\{3\}$ of $\{1,2,3\}$. Each of the 8 base graphs is combined with these 7 subsets in alternating forward/reverse order, producing $8 \times 7 = 56$ graphs.

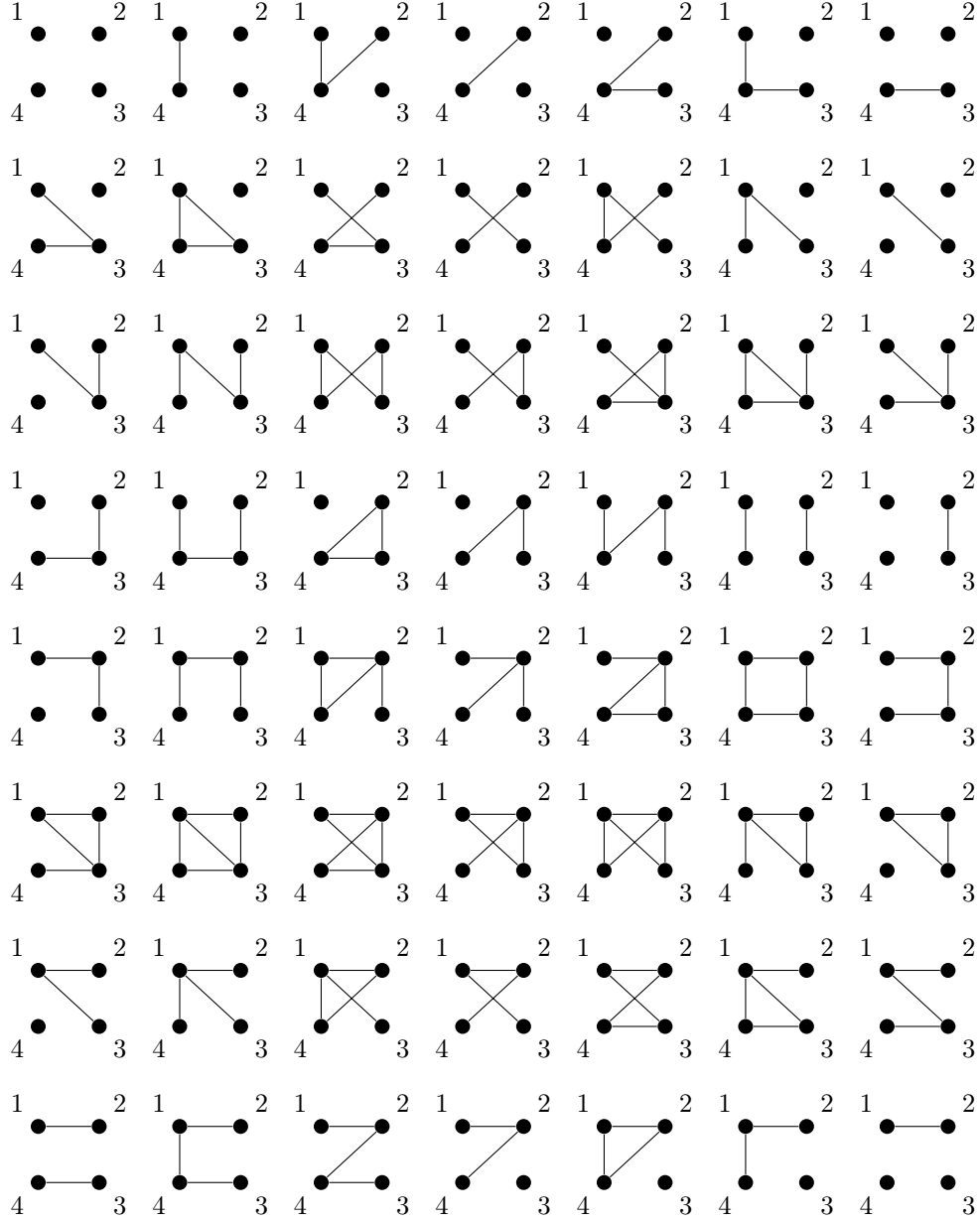
\begin{figure*}[p]
\centering
\footnotesize
\begin{tikzpicture}[
    scale=0.65,
    v/.style={circle, fill=black, inner sep=2pt},
    edge/.style={},
    lab/.style={font=\footnotesize, inner sep=3pt}
]

\newcommand{\g}[3]{
  \begin{scope}[xshift=#1cm, yshift=#2cm]
    \node[v] (1) at (0,1.15) {};
    \node[v] (2) at (1.25,1.15) {};
    \node[v] (3) at (1.25,0) {};
    \node[v] (4) at (0,0) {};
    \node[lab, above left=0.05pt of 1] {1};
    \node[lab, above right=0.05pt of 2] {2};
    \node[lab, below right=0.05pt of 3] {3};
    \node[lab, below left=0.05pt of 4] {4};
    #3
  \end{scope}
}

\g{0}{24}{}                                      
\g{2.9}{24}{\draw[edge] (1)--(4);}               
\g{5.8}{24}{\draw[edge] (1)--(4) (2)--(4);}      
\g{8.7}{24}{\draw[edge] (2)--(4);}               
\g{11.6}{24}{\draw[edge] (2)--(4) (3)--(4);}     
\g{14.5}{24}{\draw[edge] (1)--(4) (3)--(4);}     
\g{17.4}{24}{\draw[edge] (3)--(4);}              

\g{0}{20.8}{\draw[edge] (1)--(3) (3)--(4);}      
\g{2.9}{20.8}{\draw[edge] (1)--(3) (1)--(4) (3)--(4);}  
\g{5.8}{20.8}{\draw[edge] (1)--(3) (2)--(4) (3)--(4);}  
\g{8.7}{20.8}{\draw[edge] (1)--(3) (2)--(4);}    
\g{11.6}{20.8}{\draw[edge] (1)--(3) (1)--(4) (2)--(4);}  
\g{14.5}{20.8}{\draw[edge] (1)--(3) (1)--(4);}   
\g{17.4}{20.8}{\draw[edge] (1)--(3);}            

\g{0}{17.6}{\draw[edge] (1)--(3) (2)--(3);}      
\g{2.9}{17.6}{\draw[edge] (1)--(3) (1)--(4) (2)--(3);}  
\g{5.8}{17.6}{\draw[edge] (1)--(3) (1)--(4) (2)--(3) (2)--(4);}  
\g{8.7}{17.6}{\draw[edge] (1)--(3) (2)--(3) (2)--(4);}  
\g{11.6}{17.6}{\draw[edge] (1)--(3) (2)--(3) (2)--(4) (3)--(4);}  
\g{14.5}{17.6}{\draw[edge] (1)--(3) (1)--(4) (2)--(3) (3)--(4);}  
\g{17.4}{17.6}{\draw[edge] (1)--(3) (2)--(3) (3)--(4);}  

\g{0}{14.4}{\draw[edge] (2)--(3) (3)--(4);}      
\g{2.9}{14.4}{\draw[edge] (1)--(4) (2)--(3) (3)--(4);}  
\g{5.8}{14.4}{\draw[edge] (2)--(3) (2)--(4) (3)--(4);}  
\g{8.7}{14.4}{\draw[edge] (2)--(3) (2)--(4);}    
\g{11.6}{14.4}{\draw[edge] (1)--(4) (2)--(3) (2)--(4);}  
\g{14.5}{14.4}{\draw[edge] (1)--(4) (2)--(3);}   
\g{17.4}{14.4}{\draw[edge] (2)--(3);}            

\g{0}{11.2}{\draw[edge] (1)--(2) (2)--(3);}      
\g{2.9}{11.2}{\draw[edge] (1)--(2) (1)--(4) (2)--(3);}  
\g{5.8}{11.2}{\draw[edge] (1)--(2) (1)--(4) (2)--(3) (2)--(4);}  
\g{8.7}{11.2}{\draw[edge] (1)--(2) (2)--(3) (2)--(4);}  
\g{11.6}{11.2}{\draw[edge] (1)--(2) (2)--(3) (2)--(4) (3)--(4);}  
\g{14.5}{11.2}{\draw[edge] (1)--(2) (1)--(4) (2)--(3) (3)--(4);}  
\g{17.4}{11.2}{\draw[edge] (1)--(2) (2)--(3) (3)--(4);}  

\g{0}{8}{\draw[edge] (1)--(2) (1)--(3) (2)--(3) (3)--(4);}  
\g{2.9}{8}{\draw[edge] (1)--(2) (1)--(3) (1)--(4) (2)--(3) (3)--(4);}  
\g{5.8}{8}{\draw[edge] (1)--(2) (1)--(3) (2)--(3) (2)--(4) (3)--(4);}  
\g{8.7}{8}{\draw[edge] (1)--(2) (1)--(3) (2)--(3) (2)--(4);}  
\g{11.6}{8}{\draw[edge] (1)--(2) (1)--(3) (1)--(4) (2)--(3) (2)--(4);}  
\g{14.5}{8}{\draw[edge] (1)--(2) (1)--(3) (1)--(4) (2)--(3);}  
\g{17.4}{8}{\draw[edge] (1)--(2) (1)--(3) (2)--(3);}  

\g{0}{4.8}{\draw[edge] (1)--(2) (1)--(3);}      
\g{2.9}{4.8}{\draw[edge] (1)--(2) (1)--(3) (1)--(4);}  
\g{5.8}{4.8}{\draw[edge] (1)--(2) (1)--(3) (1)--(4) (2)--(4);}  
\g{8.7}{4.8}{\draw[edge] (1)--(2) (1)--(3) (2)--(4);}  
\g{11.6}{4.8}{\draw[edge] (1)--(2) (1)--(3) (2)--(4) (3)--(4);}  
\g{14.5}{4.8}{\draw[edge] (1)--(2) (1)--(3) (1)--(4) (3)--(4);}  
\g{17.4}{4.8}{\draw[edge] (1)--(2) (1)--(3) (3)--(4);}  

\g{0}{1.6}{\draw[edge] (1)--(2) (3)--(4);}      
\g{2.9}{1.6}{\draw[edge] (1)--(2) (1)--(4) (3)--(4);}  
\g{5.8}{1.6}{\draw[edge] (1)--(2) (2)--(4) (3)--(4);}  
\g{8.7}{1.6}{\draw[edge] (1)--(2) (2)--(4);}    
\g{11.6}{1.6}{\draw[edge] (1)--(2) (1)--(4) (2)--(4);}  
\g{14.5}{1.6}{\draw[edge] (1)--(2) (1)--(4);}   
\g{17.4}{1.6}{\draw[edge] (1)--(2);}            

\end{tikzpicture}
\caption{Well-ordered $2$-degenerate graphs on 4 vertices ($n=4$, $k=2$) in cyclic pivot Gray code order generated by Algorithm~\textsc{GenWOD}. Rows correspond to the 8 well-ordered $2$-degenerate graphs on 3 vertices, and columns correspond to the 7 subsets of $\{1,2,3\}$ of size at most $2$, yielding $8 \times 7 = 56$ graphs in total.}
\label{n=4,k=2}
\end{figure*}

\begin{theorem}
The algorithm \textsc{GenWOD} generates a cyclic pivot Gray code for well-ordered $k$-degenerate graphs in constant amortized time per graph and uses $O(n^2)$ space.
\end{theorem}

\begin{proof}
As standard for generation algorithms, the time required to output a graph, that is, to emit its adjacency matrix, is not part of the analysis.

By Lemma~\ref{lem:brgc-rec}, enumerating the neighbor sets of vertex $m$ once takes $O(|\mathcal{S}_{m-1,k}|)$ time. Since level $m$ performs one such enumeration for each of the $\lambda(m-1,k)$ base graphs, its work is $O(|\mathcal{S}_{m-1,k}|\cdot\lambda(m-1,k))$; as $\lambda(m,k)=|\mathcal{S}_{m-1,k}|\cdot\lambda(m-1,k)$, level $m$ costs $O(\lambda(m,k))$. For $k\ge1$ and $m\ge2$ we have $|\mathcal{S}_{m-1,k}|\ge 2$ (it contains the empty set and a singleton), so $\lambda(m,k)\ge 2\,\lambda(m-1,k)$: the number of graphs at least doubles with each level. Hence the total time is dominated by the last level and is $O(\lambda(n,k))$, which amortizes over the $\lambda(n,k)$ generated graphs to constant time per graph. For $k=0$ the only graph is the empty graph, and the claim is immediate.

For the space, the algorithm maintains only the current graph, stored as a single $n\times n$ adjacency matrix, so the total space is $O(n^2)$.
\end{proof}

\section{Enumeration}

This section presents a closed-form enumeration formula for well-ordered $k$-degenerate graphs. In contrast to the recurrence formula used by Bauer et al.~\cite{bauer2010enumerating,krugcombinatorial}, our approach provides a more explicit combinatorial representation.

\begin{theorem}\label{enu}
Let $n \ge 1$ and $0 \le k \le n-1$.
The total number of well-ordered $k$-degenerate graphs on $n$ labeled vertices is
$\prod_{v=1}^{n} \sum_{i=0}^{k} \binom{v-1}{i}$.
\end{theorem}

\begin{proof}
We count the number of $k$-degenerate graphs by following the construction process.

For each vertex $v \in \{1,2,\dots,n\}$, edges may only be added between $v$ and vertices in the set $\{1,2,\dots,v-1\}$. At step $v$, the vertex $v$ chooses a subset $S_v \subseteq \{1,2,\dots,v-1\}$ with $|S_v| \le k$. The number of such subsets is $\sum_{i=0}^{k} \binom{v-1}{i}$, since exactly $\binom{v-1}{i}$ subsets have size $i$. The choices made by different vertices are independent, so by the rule of product the total number of graphs is $\prod_{v=1}^{n} \sum_{i=0}^{k} \binom{v-1}{i}$.
\end{proof}

Theorem~\ref{enu} also yields an alternative proof of Lemma~\ref{od}.

\begin{proof}[Alternative proof of Lemma~\ref{od}]
Theorem~\ref{enu} gives $\lambda(n,k) = \prod_{v=1}^{n} \sum_{i=0}^{k} \binom{v-1}{i}$.

For $n = 1$ the product has a single factor, so $\lambda(1,k) = \binom{0}{0} = 1$. For $n \ge 2$ and $k = 0$, every factor is $\binom{v-1}{0} = 1$, hence $\lambda(n,0) = 1$. Finally, for $n \ge 2$ and $k \ge 1$, the factor at $v = 2$ equals $\binom{1}{0} + \binom{1}{1} = 2$, an even number, and as all remaining factors are positive integers, the whole product, and thus $\lambda(n,k)$, is even.
\end{proof}

\section{Ranking and Unranking}
\label{ru}

In combinatorial enumeration, ranking yields a unique index for every combinatorial object in a fixed ordered collection; unranking reconstructs the object from the index. The Gray code construction in Section~2 shows that the order of well-ordered $k$-degenerate graphs is closely related to the order of subsets in $\mathcal{S}_{n-1,k}$, which corresponds naturally to the pruned BRGC order on $\mathcal{B}_{n-1,k}$. To rank and unrank graphs in this Gray code order, we first study the ranking and unranking algorithms for binary strings of length $n$ with weight at most $k$ under the pruned BRGC order. Python implementations are provided in Appendix~B.

\subsection{Ranking and Unranking Binary Strings of Length $n$ with Weight $\le k$ under the Pruned BRGC Order}

Let $\mathcal{B}_{n,k}$ denote the set of binary strings of length $n$ with weight at most $k$, and let $C(m,K)=\sum_{i=0}^{K}\binom{m}{i}$ for $m\ge 0$ and $K\ge 0$. A simple but useful property is that ranking and unranking in the reverse Gray code order follow directly from the corresponding operations in the Gray code order. If $r_{\rm Gray}(s)$ is the rank of $s$ in the Gray code order of $\mathcal{B}_{n,k}$, then its rank in the reverse order is $|\mathcal{B}_{n,k}| - r_{\rm Gray}(s) + 1$. Hence, given a reverse rank $r_{\rm rev}$, the corresponding string is obtained by unranking at Gray code index $|\mathcal{B}_{n,k}| - r_{\rm rev} + 1$.

\subsubsection{Ranking for Binary Strings of Length $n$ with Weight $\le k$ under the Pruned BRGC Order}

Given a string $s \in \mathcal{B}_{n,k}$, we determine its rank by processing bits from right to left. For each position $i$, after fixing the suffix $s_{i+1:n} = s_{i+1}\dots s_n$, define $m_i = i-1$, $K_i = k - \sum_{j=i+1}^n s_j$, and $D_i = 1 + \sum_{j=i+1}^n s_j$; here $m_i$ is the number of bits to the left of $s_i$, $K_i$ is the number of ones that may still be placed among the first $i$ bits, and $D_i$ determines the Gray code direction at position $i$.

At each position $i$, the strings sharing the suffix $s_{i+1:n}$ partition into two blocks: those with $s_i = 1-p$ appear first, followed by those with $s_i = p$, where $p = D_i \bmod 2$. Hence, if $s_i = p$, all strings in the preceding block must be counted. Iterating over all positions, the rank accumulates as
$$
r \;=\; 1 \;+\; \sum_{i=1}^n \mathbf{1}_{\{s_i = D_i \bmod 2\}}\, C(m_i, K_i - 1 + s_i),
$$
where $C(m,K) = \sum_{t=0}^{K} \binom{m}{t}$ denotes the number of binary strings of length $m$ with weight at most $K$.

\begin{algorithm}[h]
\caption{Ranking for Binary Strings of Length $n$ with Weight $\le k$ under the Pruned BRGC Order}
\label{alg:ranking}
\begin{algorithmic}[1]
\Procedure{Ranking}{$s=s_1\dots s_n,\ k$}
    \State $m \gets n$
    \State $K \gets k$
    \State $D \gets 1$
    \State $r \gets 1$
    \For{$i \gets n, n-1, \dots, 1$}
        \State $m \gets m-1$
        \State $p \gets D \bmod 2$
        \If{$s_i = p$}
            \State $r \gets r + C(m, K - 1 + s_i)$
        \EndIf
        \State $D \gets D + s_i$
        \State $K \gets K - s_i$
    \EndFor
    \State \Return $r$
\EndProcedure
\end{algorithmic}
\end{algorithm}

Algorithm~\textsc{Ranking} implements the ranking procedure described above. Its correctness and complexity are established by Theorems~\ref{rank} and~\ref{crank}, respectively.

\begin{theorem}\label{rank}
Let $G = (s^{(1)},\dots,s^{(|\mathcal{B}_{n,k}|)})$ be the sequence of binary strings in the pruned BRGC order. Then Algorithm~\textsc{Ranking} correctly computes, for any $s \in \mathcal{B}_{n,k}$, its rank in $G$, namely
$$
r(s) \;=\; 1 \;+\; \sum_{i=1}^n \mathbf{1}_{\{s_i = D_i \bmod 2\}}\, C(m_i, K_i - 1 + s_i),
$$
where $m_i = i-1$ is the number of bits to the left of $s_i$, $K_i = k - \sum_{j=i+1}^n s_j$ is the number of ones that may still be placed among the first $i$ bits, and $D_i = 1 + \sum_{j=i+1}^n s_j$ determines the Gray code direction at position $i$.
\end{theorem}

\begin{proof}
Fix a string $s \in \mathcal{B}_{n,k}$ and a position $i$, and suppose the suffix $s_{i+1:n}$ is already fixed. Among the strings sharing this suffix, the pruned BRGC order lists first the block with $s_i = 1-p$ and then the block with $s_i = p$, where $p = D_i \bmod 2$. In the first block the $i$-th bit is fixed to $1-p$, so the number of ones still available for the remaining $m_i$ bits is $K_i - (1-p)$, which equals $K_i - 1 + s_i$ whenever $s_i = p$. Hence, if $s_i = p$, exactly $C(m_i, K_i - 1 + s_i)$ strings of the first block precede $s$, and this is exactly the quantity added to $r$ in that iteration; if $s_i = 1-p$, then $s$ itself lies in the first block, and nothing is added. Summing over all positions $i$ from $n$ down to $1$ therefore accumulates exactly the number of strings preceding $s$ in $G$, and the initial value $r = 1$ accounts for $s$ itself, so Algorithm~\textsc{Ranking} returns the rank of $s$.
\end{proof}

\begin{theorem}\label{crank}
After $O(nk)$-time preprocessing, Algorithm~\textsc{Ranking} ranks a binary string of length $n$ with weight at most $k$ in the pruned BRGC order in $O(n)$ time and $O(nk)$ space.
\end{theorem}

\begin{proof}
The values $C(m,K)=\sum_{t=0}^{K}\binom{m}{t}$ for all $m\le n$ and $K\le k$ are precomputed once, before the algorithm is invoked, in $O(nk)$ time by filling the $(n+1)\times(k+1)$ table entry by entry. The algorithm then consists of a single loop over the $n$ positions. In each iteration, it performs only the following constant number of operations: decrementing $m$, computing $p = D \bmod 2$, testing whether $s_i = p$, conditionally adding $C(m, K-1+s_i)$ to $r$ via a single table lookup, and updating $D$ and $K$. Hence every iteration takes constant time, and the overall running time is $O(n)$.

The precomputed table stores $(n+1)(k+1)$ values, which is $O(nk)$ space. The input string $s$ and the variables $m$, $K$, $D$, $r$, $i$, and $p$ contribute $O(n)$ space. The total space is therefore $O(nk)$.
\end{proof}

\medskip
We now illustrate Algorithm~\textsc{Ranking} for $n=4$ and $k=2$. The corresponding pruned BRGC order of $\mathcal{B}_{4,2}$ (the 11 strings listed in Section~2) is

\begin{center}
0000, 1000, 1100, 0100, 0110, 1010, 0010, 0011, 0101, 1001, 0001.
\end{center}

We compute the rank of $s = 1001$.

\smallskip
\noindent Initialize $m=4$, $K=2$, $D=1$, $r=1$.

\smallskip
\noindent\textbf{Position 4:} $s_4=1$.
$m=3$, $p = D \bmod 2 = 1$.
$s_4 = p$, so we add $C(3, K-1+s_4) = C(3,2) = 7$.
$r = 1+7 = 8$.

\smallskip
Update: $D = D+s_4 = 2$, $K = K-s_4 = 1$.

\smallskip
\noindent\textbf{Position 3:} $s_3=0$.
$m=2$, $p = D \bmod 2 = 0$.
$s_3 = p$, so we add $C(2, K-1+s_3) = C(2,0) = 1$.
$r = 8+1 = 9$.

\smallskip
Update: $D = D+s_3 = 2$, $K = K-s_3 = 1$.

\smallskip
\noindent\textbf{Position 2:} $s_2=0$.
$m=1$, $p = D \bmod 2 = 0$.
$s_2 = p$, so we add $C(1, K-1+s_2) = C(1,0) = 1$.
$r = 9+1 = 10$.

\smallskip
Update: $D = D+s_2 = 2$, $K = K-s_2 = 1$.

\smallskip
\noindent\textbf{Position 1:} $s_1=1$.
$m=0$, $p = D \bmod 2 = 0$.
$s_1 \neq p$, no addition. $r = 10$.

\smallskip
Thus $\mathrm{rank}(1001) = 10$.

\subsubsection{Unranking for Binary Strings of Length $n$ with Weight $\le k$ under the Pruned BRGC Order}

Given a rank $r \in \{1,\dots,|\mathcal{B}_{n,k}|\}$, we reconstruct the corresponding string $s \in \mathcal{B}_{n,k}$ by determining its bits from right to left. Let $q=r-1$ be the number of strings preceding $s$ in the pruned BRGC order, and define $m_i$, $K_i$, and $D_i$ as in the ranking procedure.

At each position $i$, the strings sharing the suffix $s_{i+1:n}$ partition into two blocks: those with $s_i = 1-p$ appear first, followed by those with $s_i = p$, where $p = D_i \bmod 2$. The first block contains $C(m_i, K_i - (1-p))$ strings. If $q \ge C(m_i, K_i - (1-p))$, then $s$ lies in the second block: we set $s_i = p$ and subtract this count from $q$; otherwise $s$ lies in the first block, and we set $s_i = 1-p$. Iterating from $i=n$ down to $1$ determines each bit $s_i$, and after the final step the string $s$ is fully reconstructed.

\begin{algorithm}[h]
\caption{Unranking for Binary Strings of Length $n$ with Weight $\le k$ under the Pruned BRGC Order}
\label{alg:unranking}
\begin{algorithmic}[1]
\Procedure{Unranking}{$n, k, r$}
    \State $s \gets 00\dots 0$
    \State $D \gets 1$
    \State $K \gets k$
    \State $q \gets r - 1$
    \For{$i \gets n, n-1, \dots, 1$}
        \State $p \gets D \bmod 2$
        \If{$q \ge C(i-1, K - 1 + p)$}
            \State $s_i \gets p$
            \State $q \gets q - C(i-1, K - 1 + p)$
        \Else
            \State $s_i \gets 1 - p$
        \EndIf
        \State $K \gets K - s_i$
        \State $D \gets D + s_i$
    \EndFor
    \State \Return $s$
\EndProcedure
\end{algorithmic}
\end{algorithm}

Algorithm~\textsc{Unranking} implements the unranking procedure described above. Its correctness and complexity are established in Theorems~\ref{unrank} and~\ref{cunrank}, respectively.

\begin{theorem}\label{unrank}
Let $G = (s^{(1)},\dots,s^{(|\mathcal{B}_{n,k}|)})$ be the sequence of binary strings in the pruned BRGC order. Then Algorithm~\textsc{Unranking} correctly returns, for any rank $r$, the unique string $s \in \mathcal{B}_{n,k}$ with rank $r$ in $G$.
\end{theorem}

\begin{proof}
Fix a position $i$, and suppose the suffix $s_{i+1:n}$ has already been determined. Among the strings sharing this suffix, the pruned BRGC order lists first the block with $s_i = 1-p$ and then the block with $s_i = p$, where $p = D_i \bmod 2$. The first block contains $C(m_i, K_i - (1-p))$ strings, and the algorithm sets $s_i = p$ exactly when $q \ge C(m_i, K_i - (1-p))$, subtracting this size from $q$ in that case. Hence, after processing position $i$, the last $n-i+1$ bits of the output form exactly the suffix of the string with rank $r$ in $G$. By induction over the positions from $n$ down to $1$, all bits are determined correctly, and since the pruned BRGC order lists every string of $\mathcal{B}_{n,k}$ exactly once, the reconstructed string is unique.
\end{proof}

\begin{theorem}\label{cunrank}
After $O(nk)$-time preprocessing, Algorithm~\textsc{Unranking} reconstructs a binary string of length $n$ with weight at most $k$ in the pruned BRGC order in $O(n)$ time and $O(nk)$ space.
\end{theorem}

\begin{proof}
After precomputing, once, the values $C(m,K)=\sum_{t=0}^{K}\binom{m}{t}$ for all $m\le n$ and $K\le k$ in $O(nk)$ time, the algorithm consists of a single loop over the $n$ positions. In each iteration, it performs a single table lookup for $C(i-1, K-1+p)$, compares $q$ with this block size, conditionally assigns $s_i$ and updates $q$, and updates the parameters $K$ and $D$; all of these take constant time. The overall running time is therefore $O(n)$.

The precomputed table has $O(nk)$ entries, and the output string $s$ together with the variables $D$, $q$, $K$, $i$, and $p$ contributes $O(n)$ space. The total space is therefore $O(nk)$.
\end{proof}

\medskip
 We illustrate Algorithm~\textsc{Unranking} on the same parameter set $n=4$, $k=2$ used in the ranking example above. Recall that the pruned BRGC order of $\mathcal{B}_{4,2}$ is
\begin{center}
0000, 1000, 1100, 0100, 0110, 1010, 0010, 0011, 0101, 1001, 0001.
\end{center}
We reconstruct the string at rank $r = 10$.

\smallskip
\noindent Initialize $D=1$, $q = r-1 = 9$, $K=2$.

\smallskip
\noindent\textbf{Position 4:} $m=3$, $p = D \bmod 2 = 1$, first block size $C(3, K-1+p) = C(3,2) = 7$.
Since $q \ge 7$, the target lies in the second block: $s_4 = p = 1$.

\smallskip
Update: $q = 2$, $K = 1$, $D = 2$.

\smallskip
\noindent\textbf{Position 3:} $m=2$, $p = D \bmod 2 = 0$, first block size $C(2, K-1+p) = C(2,0) = 1$.
Since $q \ge 1$, the target lies in the second block: $s_3 = p = 0$.

\smallskip
Update: $q = 1$, $K = 1$, $D = 2$.

\smallskip
\noindent\textbf{Position 2:} $m=1$, $p = D \bmod 2 = 0$, first block size $C(1, K-1+p) = C(1,0) = 1$.
Since $q \ge 1$, the target lies in the second block: $s_2 = p = 0$.

\smallskip
Update: $q = 0$, $K = 1$, $D = 2$.

\smallskip
\noindent\textbf{Position 1:} $m=0$, $p = D \bmod 2 = 0$, first block size $C(0, K-1+p) = C(0,0) = 1$.
Since $q < 1$, the target lies in the first block: $s_1 = 1-p = 1$.

\smallskip
Thus the algorithm returns $s = 1001$, which indeed has rank $10$ in the pruned BRGC order.

\subsection{Ranking and Unranking for well-ordered $k$-degenerate Graphs}

We now apply Algorithm~\textsc{Ranking} and Algorithm~\textsc{Unranking} for binary strings, described above, to rank and unrank well-ordered $k$-degenerate graphs in our cyclic pivot Gray code. We denote this Gray code for well-ordered $k$-degenerate graphs on $n$ vertices by $\mathbf{C}_k(n)$, and for a graph $G \in \mathbf{C}_k(n)$ we write $r(G)$ for its rank in $\mathbf{C}_k(n)$. Python implementations are provided in Appendix~C for ranking and in Appendix~D for unranking.

\subsubsection{Ranking for well-ordered $k$-degenerate Graphs}

\medskip
Given a graph $G \in \mathbf{C}_k(n)$, the following procedure computes its rank.

\begin{enumerate}
   \item \textbf{Base case:} If $n=1$, then $G$ is the trivial graph, and its rank is $1$.

    \item \textbf{Decompose the graph:} Remove vertex $n$ and its incident edges to obtain the base graph $G' = (\{1,\dots,n-1\}, \{(u,v)\in G \mid u\neq n, v\neq n\}) \in \mathbf{C}_k(n-1)$, and let $S$ be the set of neighbors of $n$, so $|S| \le k$.

   \item \textbf{Compute base graph rank:} Recursively compute $i = r(G')$, the rank of the base graph $G'$ in $\mathbf{C}_k(n-1)$.

  \item \textbf{Compute neighbor set rank:} Under the natural bijection between $\mathcal{S}_{n-1,k}$ and $\mathcal{B}_{n-1,k}$, let $b$ be the binary string corresponding to $S$, and apply Algorithm~\textsc{Ranking} to $b$; this yields the rank $j$ of $b$, taken in the forward pruned BRGC order if $i$ is odd and in the reverse order if $i$ is even, where the reverse rank is $|\mathcal{B}_{n-1,k}| - r_{\mathrm{Gray}}(b) + 1$ and $r_{\mathrm{Gray}}(b)$ denotes the rank of $b$ in the forward pruned BRGC order.

    \item \textbf{Combine ranks:} Each base graph has $M =\sum_{t=0}^{k} \binom{n-1}{t}$ possible extensions, and the rank of $G$ is $r(G) = (i-1) \cdot M + j$.
\end{enumerate}

\begin{algorithm}[h]
\caption{Ranking for well-ordered $k$-degenerate Graphs}
\label{alg:graph-ranking}
\begin{algorithmic}[1]
\Procedure{GraphRanking}{$G, n, k$}
    \If{$n = 1$}
        \State \Return 1
    \EndIf

    \State $G' \gets (\{1,\dots,n-1\}, \{(u,v)\in G \mid u\neq n, v\neq n\})$
    \State $S \gets \{ v \mid (v,n)\in G \}$
    \State $b \gets (b_1,\dots,b_{n-1})$ with $b_j \gets 1$ if $j \in S$, else $0$

    \State $i \gets \Call{GraphRanking}{G', n-1, k}$

    \If{$i$ is odd}
        \State $j \gets \Call{Ranking}{b,\ k}$
    \Else
        \State $j \gets |\mathcal{B}_{n-1,k}| - \Call{Ranking}{b,\ k} + 1$
    \EndIf

    \State $M \gets \sum_{t=0}^{k} \binom{n-1}{t}$

    \State \Return $(i-1)\cdot M + j$
\EndProcedure
\end{algorithmic}
\end{algorithm}

Algorithm~\textsc{GraphRanking} implements the recursive ranking procedure
described above. Its correctness and complexity are established by
Theorems~\ref{grank} and~\ref{cgrank}, respectively.

\begin{theorem}\label{grank}
The algorithm \textsc{GraphRanking} correctly computes the rank of any well-ordered $k$-degenerate graph $G$ in the cyclic pivot Gray code.
\end{theorem}

\begin{proof}
We prove by induction on $n$. For $n=1$ the only graph is the trivial graph, and the algorithm returns $1$, which is correct. Now assume the algorithm correctly ranks all graphs in $\mathbf{C}_k(n-1)$, and let $G \in \mathbf{C}_k(n)$ with base graph $G'$ and neighbor set $S$, $|S|\le k$. By the induction hypothesis, $i = r(G')$ is the rank of $G'$ in $\mathbf{C}_k(n-1)$. In $\mathbf{C}_k(n)$ the extensions of $G'$ occupy the consecutive positions $(i-1)M+1,\dots,iM$. Within this block, the subsets are listed in forward pruned BRGC order when $i$ is odd and in reverse order when $i$ is even; correspondingly, $j$ is computed by Algorithm~\textsc{Ranking} in the forward order when $i$ is odd, and by the reverse-order formula $j = |\mathcal{B}_{n-1,k}| - r_{\mathrm{Gray}}(b) + 1$ when $i$ is even, where $r_{\mathrm{Gray}}(b)$ denotes the rank of $b$ in the forward pruned BRGC order. Hence $(i-1)\cdot M + j$ is the position of $G$ in $\mathbf{C}_k(n)$.
\end{proof}

\begin{theorem}\label{cgrank}
The algorithm~\textsc{GraphRanking} ranks well-ordered $k$-degenerate graphs on $n$ vertices in $O(n^2)$ time and uses $O(n^2)$ space.
\end{theorem}

\begin{proof}
There are $n$ recursion levels, one per vertex from $n$ down to $1$. Consider level $m$, where the graph has $m$ vertices. Its cost consists of three parts: reading the $m$-th row of the adjacency matrix to decompose $G$ costs $O(m)$; ranking the neighbor set, which is a binary string of length $m-1$, costs $O(m-1)$ by Theorem~\ref{crank}; and everything else---the parity check, the computation of $M$, and the final combination $(i-1)\cdot M+j$---costs $O(1)$. The table used by Algorithm~\textsc{Ranking} is precomputed once and shared by all levels, at a total cost of $O(nk)\le O(n^2)$, which is dominated by the bound. Thus level $m$ costs $O(m)$, and summing over all levels gives $\sum_{m=1}^{n} O(m) = O(n^2)$.

The algorithm stores the input graph as a single $n\times n$ adjacency matrix, contributing $O(n^2)$ space, and the recursion passes only the base graph together with a constant number of parameters between levels, so no additional copies of the graph are ever created. Hence the total space complexity is $O(n^2)$.
\end{proof}

\smallskip

We now illustrate Algorithm~\textsc{GraphRanking} on the example from Figure~\ref{fig2}.

\begin{example}
Let $n=3$, $k=2$. Consider the graph $G$ with vertex set $\{1,2,3\}$ and edge set $\{(1,2), (1,3)\}$.

\begin{itemize}
    \item Base graph: $G' = (\{1,2\}, \{(1,2)\}) \in \mathbf{C}_2(2)$

    Rank of $G'$: $i=2$ (even).

    \item Neighbor set of vertex $3$: $S = \{1\}$. The pruned BRGC order of length $2$ is $00,10,11,01$, corresponding to subsets $\varnothing,\{1\},\{1,2\},\{2\}$. Since $i$ is even, the reverse order is used: $\{2\},\{1,2\},\{1\},\varnothing$. Hence $S = \{1\}$ is at position $j = 3$.

    \item Number of subsets: $M = \sum_{t=0}^{2} \binom{2}{t} = 1+2+1 = 4$.

    \item Total rank: $r(G) = (i-1) \cdot M + j = (2-1)\cdot 4 + 3 = 7$.
\end{itemize}

Thus the graph $G$ is ranked 7th in $\mathbf{C}_2(3)$.
\end{example}

\subsubsection{Unranking for well-ordered $k$-degenerate Graphs}

\medskip
Given a rank $x \in \{1,\dots,|\mathbf{C}_k(n)|\}$, the following procedure reconstructs the corresponding graph.

\begin{enumerate}
    \item \textbf{Base case:} If $n=1$, return the trivial graph $(\{1\}, \varnothing)$.

    \item \textbf{Determine group size:} For $n \ge 2$, each base graph in $\mathbf{C}_k(n-1)$ has $M = \sum_{t=0}^{k} \binom{n-1}{t}$ possible extensions, obtained by connecting vertex $n$ to a subset of $\{1,\dots,n-1\}$ of size at most $k$.

    \item \textbf{Decompose index:} Compute $i = \left\lfloor \frac{x-1}{M} \right\rfloor + 1$ and $j = x - (i-1)\cdot M$, so that $i$ is the index of the base graph and $j$ is the index of the neighbor set within its block.

    \item \textbf{Reconstruct base graph:} Recursively reconstruct the base graph $G'$, namely the graph with rank $i$ in $\mathbf{C}_k(n-1)$.

    \item \textbf{Reconstruct neighbor set:} If $i$ is odd, unrank $j$ in the forward pruned BRGC order by Algorithm~\textsc{Unranking}; if $i$ is even, unrank in the reverse order, that is, apply Algorithm~\textsc{Unranking} to the index $|\mathcal{B}_{n-1,k}| - j + 1$ in the forward order.

    \item \textbf{Construct final graph:} Set $E = E(G') \cup \{(v,n) \mid v \in S\}$ and $G = (\{1,\dots,n\}, E)$.
\end{enumerate}

\begin{algorithm}[h]
\caption{Unranking for well-ordered $k$-degenerate Graphs}
\label{alg:graph-unranking}
\begin{algorithmic}[1]
\Procedure{GraphUnranking}{$x, n, k$}
    \If{$n = 1$}
        \State \Return $(\{1\}, \varnothing)$
    \EndIf

    \State $M \gets \sum_{t=0}^{k} \binom{n-1}{t}$
    \State $i \gets \left\lfloor \dfrac{x-1}{M} \right\rfloor + 1$
    \State $j \gets x - (i-1)\cdot M$

    \State $G' \gets \Call{GraphUnranking}{i, n-1, k}$

    \If{$i$ is odd}
        \State $b \gets \Call{Unranking}{n-1, k, j}$
    \Else
        \State $b \gets \Call{Unranking}{n-1, k, |\mathcal{B}_{n-1,k}| - j + 1}$
    \EndIf
    \State $S \gets \{ v \mid b_v = 1 \}$

    \State $E \gets E(G') \cup \{(v,n) \mid v \in S\}$
    \State \Return $(\{1,\dots,n\}, E)$
\EndProcedure
\end{algorithmic}
\end{algorithm}

Algorithm~\textsc{GraphUnranking} implements the recursive unranking procedure
described above. Its correctness and complexity are established by
Theorems~\ref{gunrank} and~\ref{cgunrank}, respectively.

\begin{theorem}\label{gunrank}
The algorithm~\textsc{GraphUnranking} correctly reconstructs the well-ordered $k$-degenerate graph in $\mathbf{C}_k(n)$ with rank $x$ in the cyclic pivot Gray code.
\end{theorem}

\begin{proof}
We prove by induction on $n$. For $n=1$ the only graph is the trivial graph, and the algorithm returns it. Now assume correct reconstruction for $n-1$. In $\mathbf{C}_k(n)$ each base graph occupies $M = \sum_{t=0}^{k} \binom{n-1}{t}$ consecutive positions, so the index $x$ determines the base graph index $i = \left\lfloor \frac{x-1}{M} \right\rfloor + 1$ and the offset $j = x - (i-1)\cdot M$ within the block. By the induction hypothesis, the recursive call reconstructs the graph with rank $i$ in $\mathbf{C}_k(n-1)$, and since the subsets are listed in forward pruned BRGC order when $i$ is odd and in reverse order when $i$ is even, Algorithm~\textsc{Unranking} applied to $j$ in the forward order, or to $|\mathcal{B}_{n-1,k}| - j + 1$ in the reverse case, reconstructs the correct neighbor set $S$. Combining $G'$ and $S$ yields the unique graph with rank $x$ in $\mathbf{C}_k(n)$.
\end{proof}

\begin{theorem}\label{cgunrank}
The algorithm~\textsc{GraphUnranking} reconstructs well-ordered $k$-degenerate graphs on $n$ vertices in $O(n^2)$ time and uses $O(n^2)$ space.
\end{theorem}

\begin{proof}
There are $n$ recursion levels, one per vertex from $n$ down to $1$. Consider level $m$, where the graph has $m$ vertices. Its cost consists of the following parts: computing the quotient $i$ and the remainder $j$ from $x$ costs $O(1)$; the call to Algorithm~\textsc{Unranking} on a binary string of length $m-1$ costs $O(m-1)$ by Theorem~\ref{cunrank}; and the reconstruction of $S$ from the string together with the union of the edge sets costs $O(m)$. The table used by Algorithm~\textsc{Unranking} is precomputed once and shared by all levels, at a total cost of $O(nk)\le O(n^2)$, which is dominated by the bound. Thus level $m$ costs $O(m)$, and summing over all levels gives $\sum_{m=1}^{n} O(m) = O(n^2)$.

The algorithm constructs a single $n\times n$ adjacency matrix, contributing $O(n^2)$ space, and the recursion passes only a constant number of parameters between levels, so no additional copies of the graph are ever created. Hence the total space complexity is $O(n^2)$.
\end{proof}

We now illustrate Algorithm~\textsc{GraphUnranking} on the same parameter set $n=3$, $k=2$ with rank $r(G)=7$.

\begin{example}
Let $n=3$, $k=2$, and rank $r(G)=7$. Reconstruct the graph $G$:

\begin{itemize}
    \item Compute $M = \sum_{t=0}^{2} \binom{2}{t} = 4$.
    \item Decompose index: $i = \lfloor (7-1)/4 \rfloor + 1 = 2$, $j = 7 - (2-1)\cdot 4 = 3$.
    \item Base graph $G' = (\{1,2\}, \{(1,2)\}) \in \mathbf{C}_2(2)$.
    \item Neighbor set of vertex $3$: $S$. Since $i=2$ is even, use the reverse Gray code of subsets of $\{1,2\}$ of size $\le 2$ in the correct order: $\{2\}, \{1,2\}, \{1\}, \varnothing$, so $S = \{1\}$ is at position $j=3$.
    \item Construct edges: $E = \{(1,2)\} \cup \{(1,3)\} = \{(1,2),(1,3)\}$.
    \item Return graph: $G = (\{1,2,3\}, E)$.
\end{itemize}
\end{example}

\section*{Acknowledgements}
This research is supported by the Macao Polytechnic University research grant (Project code: RP/FCA-01/2026 and fca.ec0b.90da.3).

\section*{Data availability}
No data was used for the research described in the article.

\newpage
\bibliographystyle{elsarticle-num}
\bibliography{references/myrefs}

\newpage
\appendix
\section{Python code to generate well-ordered $k$-degenerate graphs}\label{sec:pyccode}
{
\tiny
\lstset{style=mystyle}
\begin{lstlisting}[language=python]
def gen_brgc_k(n, k, forward=True, toggle=None):
    bits = [0] * n
    weight = 0

    def visit(i, fwd):
        nonlocal weight
        if i == -1 or weight == k:
            yield tuple(bits)
            return
        if fwd:
            yield from visit(i - 1, True)
            bits[i] = 1
            weight += 1
            if toggle:
                toggle(i, 1)
            yield from visit(i - 1, False)
            bits[i] = 0
            weight -= 1
            if toggle:
                toggle(i, 0)
        else:
            bits[i] = 1
            weight += 1
            if toggle:
                toggle(i, 1)
            yield from visit(i - 1, True)
            bits[i] = 0
            weight -= 1
            if toggle:
                toggle(i, 0)
            yield from visit(i - 1, False)

    yield from visit(n - 1, forward)


def GENWOD(n, k, visit=lambda A: None):
    A = [[0] * n for _ in range(n)]

    def rec(m):
        if m == 1:
            yield None
            return
        d = 0

        def toggle(i, v):
            A[m - 1][i] = A[i][m - 1] = v

        for _ in rec(m - 1):
            yield from gen_brgc_k(m - 1, k, forward=(d == 0), toggle=toggle)
            d = 1 - d

    for _ in rec(n):
        visit(A)


def format_print_matrix(matrix):
    for row in matrix:
        print("    " + "  ".join(str(val) for val in row))


def main():
    print("=== Well-ordered k-degenerate Graph Generator ===")
    try:
        n_input = int(input("Please enter the number of vertices n: "))
        k_input = int(input("Please enter the maximum degree k: "))
    except ValueError:
        print("Invalid input! Please enter valid integers.")
        return

    print("\nGenerating, please wait...")
    total = [0]

    def snapshot(A):
        total[0] += 1
        print(f"\n[Graph {total[0]}]")
        format_print_matrix(A)

    GENWOD(n_input, k_input, snapshot)

    print("-" * 45)
    print(f"Generation complete! n={n_input}, k={k_input}")
    print(f"Total graphs generated: {total[0]}")
    print("-" * 45)


if __name__ == "__main__":
    main()
\end{lstlisting}
}

\clearpage
\section{Python code to rank and unrank binary strings of length $n$ with weight $\le k$ in Gray code order}
{
\tiny
\lstset{style=mystyle}
\begin{lstlisting}[language=python]
import math

def build_table(n, k):
    C = [[0] * (k + 1) for _ in range(n + 1)]
    for m in range(n + 1):
        C[m][0] = 1
        for K in range(1, k + 1):
            C[m][K] = C[m][K - 1] + (math.comb(m, K) if K <= m else 0)
    return C

def count(C, m, K):
    if K < 0:
        return 0
    if K > len(C[0]) - 1:
        K = len(C[0]) - 1
    return C[m][K]

def ranking(n, k, s):
    C = build_table(n, k)
    r = 1
    d = 1
    for s_i in reversed(s):
        n -= 1
        r += count(C, n, k - 1 + s_i) if s_i == d % 2 else 0
        d += s_i
        k -= s_i
    return r

def unranking(n, k, r):
    C = build_table(n, k)
    s = [0] * n
    d = 1
    r -= 1
    for i in range(n - 1, -1, -1):
        if r >= count(C, i, k - 1 + d % 2):
            s[i] = d % 2
            r -= count(C, i, k - 1 + d % 2)
        else:
            s[i] = 1 - d % 2
        k -= s[i]
        d += s[i]
    return s

mode = input("Enter 1 for Ranking, 0 for Unranking: ").strip()

if mode == "1":
    n = int(input("Enter binary string length n: ").strip())
    k = int(input("Enter weight limit k: ").strip())
    s_str = input("Enter binary string (length n): ").strip()
    if len(s_str) != n or any(c not in '01' for c in s_str):
        print("Invalid string length or format!")
    else:
        s = [int(c) for c in s_str]
        r = ranking(n, k, s)
        print("Ranking result: ", r)
elif mode == "0":
    n = int(input("Enter binary string length n: ").strip())
    k = int(input("Enter weight limit k: ").strip())
    r = int(input("Enter index r (starting from 1): ").strip())
    C = build_table(n, k)
    total = count(C, n, k)
    if not (1 <= r <= total):
        print(f"Index r out of range! Valid range: 1 ~ {total}")
    else:
        s = unranking(n, k, r)
        s_str = "".join(str(x) for x in s)
        print("Unranking result: ", s_str)
else:
    print("Invalid input, please enter 1 or 0")


\end{lstlisting}
}

\clearpage
\section{Python code for ranking well-ordered $k$-degenerate graphs}\label{sec:pyccode}
{
\tiny
\lstset{style=mystyle}
\begin{lstlisting}[language=python]
import math

def build_table(n, k):
    C = [[0] * (k + 1) for _ in range(n + 1)]
    for m in range(n + 1):
        C[m][0] = 1
        for K in range(1, k + 1):
            C[m][K] = C[m][K - 1] + (math.comb(m, K) if K <= m else 0)
    return C

def count(C, m, K):
    if K < 0:
        return 0
    if K > len(C[0]) - 1:
        K = len(C[0]) - 1
    return C[m][K]

def ranking(C, k, s):
    n = len(s)
    r = 1
    d = 1
    for s_i in reversed(s):
        n -= 1
        r += count(C, n, k - 1 + s_i) if s_i == d % 2 else 0
        d += s_i
        k -= s_i
    return r

def graph_ranking(G, n, k, C=None):
    if C is None:
        C = build_table(n, k)
    if n == 1:
        return 1
    s = {i for i, val in enumerate(G[n-1][:n-1]) if val == 1}
    for i in s:
        G[i][n-1] = G[n-1][i] = 0
    s = [1 if i in s else 0 for i in range(n-1)]
    i = graph_ranking(G, n-1, k, C)
    if i % 2 != 0:
        j = ranking(C, k, s)
    else:
        j = count(C, n-1, k) - ranking(C, k, s) + 1
    m = count(C, n-1, k)
    return (i-1) * m + j

def main():
    print("=== Graph Ranking Input/Output ===")
    try:
        n = int(input("Enter n (number of vertices): "))
        k = int(input("Enter k (degree limit): "))
        print(f"Enter the {n}x{n} adjacency matrix row by row (space-separated numbers):")
        G = []
        for _ in range(n):
            row_input = input().strip().split()
            row = [int(x) for x in row_input]
            if len(row) != n:
                print(f"Error: Expected {n} elements in the row, got {len(row)}.")
                return
            G.append(row)
        rank = graph_ranking(G, n, k)
        print("\n--------------------------------")
        print(f"Calculated Rank: {rank}")
        print("--------------------------------")
    except ValueError:
        print("Invalid input. Please enter valid integers.")
    except Exception as e:
        print(f"An unexpected error occurred: {e}")

if __name__ == "__main__":
    main()
\end{lstlisting}
}

\clearpage
\section{Python code for unranking well-ordered $k$-degenerate graphs}\label{sec:pyccode}
{
\tiny
\lstset{style=mystyle}
\begin{lstlisting}[language=python]
import math

def build_table(n, k):
    C = [[0] * (k + 1) for _ in range(n + 1)]
    for m in range(n + 1):
        C[m][0] = 1
        for K in range(1, k + 1):
            C[m][K] = C[m][K - 1] + (math.comb(m, K) if K <= m else 0)
    return C

def count(C, m, K):
    if K < 0:
        return 0
    if K > len(C[0]) - 1:
        K = len(C[0]) - 1
    return C[m][K]

def unranking(C, n, k, r):
    s = [0] * n
    d = 1
    r -= 1
    for i in range(n - 1, -1, -1):
        if r >= count(C, i, k - 1 + d % 2):
            s[i] = d % 2
            r -= count(C, i, k - 1 + d % 2)
        else:
            s[i] = 1 - d % 2
        k -= s[i]
        d += s[i]
    return s

def graph_unranking(x, n, k, G=None, C=None):
    if C is None:
        C = build_table(n, k)
    if G is None:
        G = [[0] * n for _ in range(n)]
    if n == 1:
        return G
    m = count(C, n - 1, k)
    i = (x - 1) // m + 1
    j = x - (i - 1) * m
    G = graph_unranking(i, n - 1, k, G, C)
    if i % 2 != 0:
        s = unranking(C, n - 1, k, j)
    else:
        s = unranking(C, n - 1, k, count(C, n - 1, k) - j + 1)
    for idx, bit in enumerate(s):
        if bit == 1:
            G[idx][n - 1] = 1
            G[n - 1][idx] = 1
    return G

def format_matrix(matrix):
    if not matrix:
        return "[]"
    formatted_rows = []
    for row in matrix:
        formatted_rows.append("  " + " ".join(str(val) for val in row))
    return "\n".join(formatted_rows)

def main():
    print("=== Graph Unranking Input/Output ===")
    try:
        n = int(input("Enter n (number of vertices): "))
        k = int(input("Enter k (degree limit): "))
        rank = int(input("Enter the target Rank: "))
        matrix = graph_unranking(rank, n, k)
        print("\n--------------------------------")
        print("Generated Graph Matrix:")
        print(format_matrix(matrix))
        print("--------------------------------")
    except ValueError:
        print("Invalid input. Please enter valid integers.")
    except Exception as e:
        print(f"An unexpected error occurred: {e}")

if __name__ == "__main__":
    main()
\end{lstlisting}
}



\end{document}